\documentclass[aps,prx,reprint,nofootinbib]{revtex4-1}
\usepackage{blindtext}

\usepackage{packages}
\newcommand{\normord}[1]{:\mathrel{#1}:}

\renewcommand{\normord}[1]{:\mathrel{#1}:}
\usepackage{animate}
\newcommand{\acom}[2]{\left\{#1,#2\right\}}
\newcommand{\px}{\ket{x}\bra{x}}
\newcommand{\sO}{\mathcal{O}}
\newcommand{\sF}[1]{\mathcal{F}\of{\sH_{#1}}}
\renewcommand{\P}{\mathbb{P}}

\begin{document}
\title{How to Partition a Quantum Observable}
\author{Caleb M.\ Webb}
%\email{mail@example.com}
\author{Charles A.\ Stafford}
\affiliation{Department of Physics, University of Arizona, 1118 E.\ 4th Street, Tucson, AZ 85721}

\date{\today}

\begin{abstract}
We present a partition of quantum observables in an open quantum system which is inherited from the division of the underlying Hilbert space or configuration space. It is shown that this partition leads to the definition of an inhomogeneous continuity equation for generic, non-local observables. This formalism is employed to describe the local evolution of the von Neumann entropy of a system of independent quantum particles out of equilibrium. Crucially, we find that all local fluctuations in the entropy are governed by an entropy current operator, implying that the production of entanglement entropy is not measured by this partitioned entropy. For systems linearly perturbed from equilibrium, it is shown that this entropy current is equivalent to a heat current, provided that the system-reservoir coupling is partitioned symmetrically. Finally, we show that any other partition of the coupling leads directly to a divergence of the von Neumann entropy. Thus, we conclude that Hilbert-space partitioning is the only partition of the von Neumann entropy which is consistent with the Laws of Thermodynamics.
\end{abstract}

\maketitle

%\blindtext \cite{article-minimal}

%\bibliographystyle{apsrev4-1} % Tell bibtex which bibliography style to use
%\bibliography{xampl} % Tell bibtex which .bib file to use (this one is some example file in %TexLive's file tree)

\section{Introduction}
Consider a global Hilbert space $\sH$, which contains all single-particle degrees of freedom available to the universe. From this we construct the Fock space $\mathcal{F}\of{\sH}$, within which all many-body theories are described. In the study of open quantum systems, one typically divides the universe into complementary sets of orthogonal, single-particle states $\sH=\sH_S\oplus\sH_R$ describing the subsystem of interest $\sH_S$ and a thermal reservoir $\sH_R$. The Fock space then decomposes as the familiar product $\mathcal{F}\of{\sH}\cong\mathcal{F}\of{\sH_S}\otimes\mathcal{F}\of{\sH_R}$.

Denote by $\sL\of{\sF{}}$ the set of linear operators acting on $\mathcal{\sF{}}$. Any observable $\hat{\sO}\in\sL\of{\sF{}}$ may, of course, be written in the form $\hat{\sO}=\hat{\sO}_S+\hat{\sO}_R+\hat{\sO}_{SR}$, where $\hat{\sO}_{S/R}\in\sL\of{\sF{S/R}}$ and $\hat{\sO}_{SR}$ describes the coupling between the two subsystems. One may then ask whether the expectation value $\braket{\hat{\sO}}$ can be sensibly divided into system and reservoir contributions. In particular, what fraction of the averaged coupling term $\braket{\hat{\sO}_{SR}}$ should be assigned to each subsystem? For a partitioning of $\sH$ into subsystems of comparable size it seems intuitive that the coupling should be partitioned symmetrically, assigning $\frac{1}{2}\braket{\hat{\sO}_{SR}}$ to each subsystem. However, in the context of the thermodynamics of open quantum systems, an answer to this question has remained elusive, due to subtleties in the distinction between system internal energy and heat \cite{Ludovico,EOG,bruchQuantumThermodynamicsDriven2016,haughianQuantumThermodynamicsResonantlevel2018,talknerColloquiumStatisticalMechanics2020,strasbergFirstSecondLaw2021a,bergmannGreenFunctionPerspective2021,lacerdaQuantumThermodynamicsFast2023}. %[TO DO: more refs?] 
This thermodynamic bookkeeping typically falls into one of two camps: either half of the coupling Hamiltonian is assigned to the system, as in the symmetric partition, or all of it is. The remaining fraction is then assigned to the reservoir and used to describe heat dissipated. 

In this article, we maintain that any partition of the energetics between two subsystems should be inherited from the division of the single-particle states $\sH\rightarrow\sH_S\oplus\sH_R$. Moreover, it will be shown that such a ``Hilbert-space partition" of any observable leads directly to the symmetric partition.

Of central interest in the study of open quantum systems is the von Neumann entropy $S=-\Tr{\hat{\rho}\ln{\hat{\rho}}}$, where $\hat{\rho}$ is the density matrix\footnote{$\hbar=k_B=1$ throughout this paper.} 
describing the state of the global system in $\sF{}$. A description of the \textit{system entropy} may be constructed by means of the reduced density matrix, defined from $\hat{\rho}$  by tracing out the reservoir degrees of freedom, $\hat{\rho}^{red}_S=\Tr_R{\hat{\rho}}$. The reduced entropy of the system is then defined as $S^{red}_S=-\Tr{\hat{\rho}^{red}_S\ln{\hat{\rho}^{red}_S}}$. Using a quantum master equation to describe the dynamics of the reduced density matrix for the system, it has been shown \cite{Esposito_2010,Potts_2021} that the reduced entropy obeys
\begin{align} \label{Sredcons}
\frac{d}{dt}S^{red}_S=\frac{J_Q}{T}+\frac{d}{dt}\Sigma,
\end{align}
where $J_Q$ is a heat current flowing between the system and reservoir and $\Sigma(t)$ is the entanglement entropy between the subsystems. While interesting from the point of view of quantum information theory, this production of entanglement entropy, we argue, is not thermodynamic in character. 

As a simple illustration of our argument, consider an infinite, fermionic, tight-binding chain in thermal equilibrium at temperature $T$.
\begin{align}
\hat{H}&=-t\SUM{n}{}\hat{c}^{\dagger}_n\hat{c}_{n+1}+\mbox{h.c.},\\
\hat{\rho}&=\frac{e^{-\beta\hat{H}}}{\Tr{e^{-\beta\hat{H}}}}.
\end{align}
Here we have taken the chemical potential to lie in the center of the electron band. With the chain in thermal equilibrium, one might expect that all extensive quantities should be uniformly distributed throughout the chain. In particular, for the entropy of a single site we anticipate that $S=\frac{1}{N}S_{total}$, where $N$ is the number of sites in the chain and
\begin{align} \label{stot}
S_{total}&=-\SUM{k}{}[f_k\ln{f_k}+(1-f_k)\ln{1-f_k}],
\end{align}
with $f_k=\frac{1}{e^{\beta\epsilon_k}+1}$ the Fermi-Dirac distribution for orbital $\ket{k}$ at temperature $T=\beta^{-1}$. On the other hand, the reduced state for a single site will be of the form $\hat{\rho}^{red}=f\hat{n}+(1-f)(1-\hat{n})$ where $\hat{n}$ is the number operator for the site and $f$ gives the probability that this site is occupied. Using the results presented in Ref.\ \cite{Dhar}, it can be shown that
\begin{align}
f=\frac{1}{N}\SUM{k}{}f_k=1/2,
\end{align}
which is to be expected for the occupancy of a single site given $\mu$ in the center of the band, regardless of temperature. The reduced entropy is thus 
\begin{align} \label{Sred} 
S^{red}=-f\ln{f}-(1-f)\ln{1-f}=\ln{2},
\end{align}
independent of temperature. Such an entropy is to be expected for the chain's classical counterpart in a microcanonical ensemble, with $N$ sites filled by $N/2$ particles. The classical picture neglects the entanglement between the fermions occupying the chain; though non-interacting, the antisymmetric nature of the fermionic wavefunctions enforces a maximal entanglement for all states in the ensemble. Indeed, upon comparing \cref{stot} and \cref{Sred}, we see that $S-S^{red}=\frac{\Sigma}{N}$, with $\Sigma$ the total entanglement entropy of the lattice sites. 

In particular, as $T\rightarrow0$ the chain approaches a pure state with all particles maximally entangled. Thus, we should have that $\lim{T}{0}S_{total}(T)=0$, in accordance with the Third Law of Thermodynamics. A comparison between the entropy density $S_{total}/N$ and reduced entropy is shown in \cref{F1}.

\begin{figure}[t]
\begin{tikzpicture}
\node at (0,0) {\includegraphics[width=0.5\textwidth]{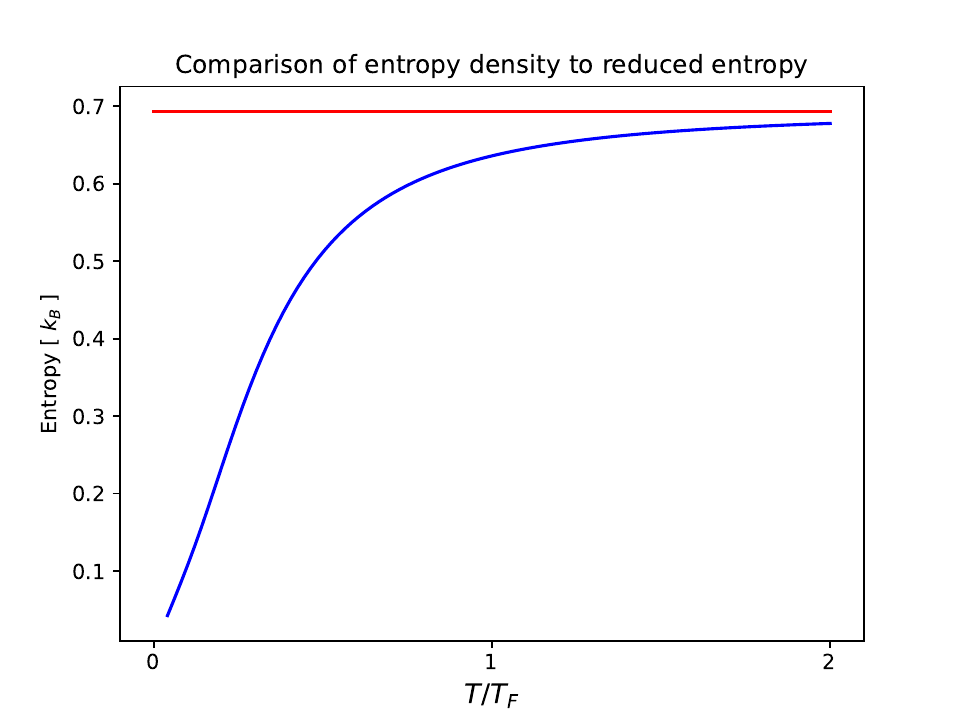}};
\draw[stealth-stealth] (-2.54,-1.17) -- (-2.54,2.315) node[pos=0.5,left] {$\frac{\Sigma}{N}$};
\end{tikzpicture}
\caption{Comparison of entropy density to reduced entropy of a single site in an infinite tight-binding chain at thermal equilibrium and chemical potential $\mu=0$. The entropy density clearly satisfies the Third Law of Thermodynamics, while the reduced entropy is independent of temperature. The difference between the two entropies is given exactly by the entanglement entropy per site of the lattice.} \label{F1}
\end{figure}

Clearly, the reduced entropy does not have a meaningful thermodynamic interpretation, whereas the entropy density does, at least for this toy model. In what follows, it will be shown that the Hilbert-space partition can be employed to give a local description of the entropy which \textit{does} have a meaningful thermodynamic interpretation. Importantly, it will be shown that this partitioned entropy does not measure the entanglement between the subsystems.

\section{Hilbert-Space Partition}
In what follows, we will denote by $\hat{\sO}|_S$ the partition of observable $\hat{\sO}$ over the subspace $\sH_S\subset\sH$. We seek a partition which is
\begin{enumerate}
\item Inherited from the division of single-particle states on $\sH=\sH_S\oplus\sH_R$, and
\item Additive over subspaces: $\braket{\hat{\sO}}=\braket{\hat{\sO}|_S}+\braket{\hat{\sO}|_R}$.
\end{enumerate}

For the sake of clarity, we focus for now on one-body observables $\hat{\sO}=\SUM{nm}{}O_{nm}\hat{c}_n^{\dagger}\hat{c}_m\in\sL\of{\sF{}}$, and discuss the generalizations to $N$-body observables in Sec.\ \ref{sec:many_body}.
The matrix elements $O_{nm}$ can be extracted in the form of an observable in first quantization, $O\in\sL\of{\sH}$, which acts on the single-particle Hilbert space. As it is this space which is being divided into system and reservoir states, the partition should be constructed at the level of these matrix elements. By second quantization of the partitioned operator $O|_S$ we recover $\hat{\sO}|_S\in\sL\of{\sF{S}\otimes\sF{R}}$. This construction is outlined in \cref{fig2}.

For the partition of $O$, we define
\begin{align} \label{partition}
O|_S\equiv\frac{1}{2}\acom{O}{\mathbb{P}_S},
\end{align}
where $\P_S=\SUM{\phi_j\in\sH_S}{}\ket{\phi_j}\bra{\phi_j}$ is an orthogonal projector onto $\sH_S$ and the anti-commutator $\acom{\cdot}{\cdot}$ is included to ensure that $O|_S$ is Hermitian. 
The Fock-space operator $\hat{\sO}|_S$ is then given by
\begin{align} \label{hatpartition}
\hat{\sO}|_S\equiv\SUM{nm}{}\bra{n}O|_S\ket{m}\hat{c}^{\dagger}_nc_m.
\end{align}
That this partition satisfies the additivity condition is a simple consequence of the fact that $\P_S+\P_R=\id$, so that $\hat{\sO}|_S+\hat{\sO}|_R=\hat{\sO}$. 

Furthermore, defining
\begin{align}
\begin{aligned}
    &O_S=\P_SO\P_S,\\
    &O_{SR}=\P_SO\P_R+\P_RO\P_S,\\
    &O_R=\P_RO\P_R,
    \end{aligned}
\end{align}
we see that $O|_S=O_S+\frac{1}{2}O_{SR}$. Performing second quantization, one therefore finds that this partition is equivalent to the symmetric partition, wherein the system-reservoir coupling is partitioned equally between subsystems
\begin{equation}\label{eq:partition_sym}
 \hat{\sO}|_S=\hat{\sO}_S+\frac{1}{2}\hat{\sO}_{SR}.   
\end{equation}
 In the remainder of this paper, we will use $\hat{H}$ to denote the second quantization of an operator $H\in\sL\of{\sH}$.

\begin{figure}[t]
\centering
\resizebox{\columnwidth}{!}{%
\begin{tikzpicture}
\draw[-latex] (-3,1) node[anchor=east] (D) {$\mathcal{F}\of{\sH}$} --++ (1,0) node[anchor=west] (A) {$\sH$};
\draw[-latex] (A.east) --++ (1,0) node[anchor=west] (B) {$\sH_S\oplus\sH_R$};
\draw[-latex] (B.east) --++ (1,0) node[anchor=west] (C) {$\mathcal{F}\of{\sH_S}\otimes\mathcal{F}\of{\sH_R}$};

\path (D.east) --++ (0,-1.5) node (DR) {};
\path (A.west) --++ (0,-1.5) node (AL) {};
\path (A.east) --++ (0,-1.5) node (AR) {};
\path (B.west) --++ (0,-1.5) node (BL) {};
\path (B.east) --++ (0,-1.5) node (BR) {};

\draw[-latex] (DR.center) -- (AL.center) node[anchor=west] {$O$};
\draw[-latex] (AR.center) -- (BL.center);
\draw[-latex] (BR.center) --++ (1,0); 

\path (D.center) --++ (0,-1.5) node {$\hat{\mathcal{O}}$};
\path (B.center) --++ (0,-1.5) node {$O|_S$};
\path (C.center) --++ (0,-1.5) node {$\hat{\mathcal{O}}|_S$};

\end{tikzpicture}
}
\caption{The general scheme for the partitioning of a Fock-space operator $\hat{\sO}$. First, the matrix $O$ in first quantization is constructed from the matrix elements of $\hat{\sO}$. A partition is then constructed for $O\rightarrow O|_S$. Finally, this operator is second-quantized to form the partitioned operator, acting on the partitioned Fock Space.}\label{fig2}
\end{figure}

\section{Time Dependence}
In order to construct the local dynamics of a partitioned expectation value $\braket{\hat{\sO}|_S}(t)$, it will be useful to consider the Heisenberg evolution of the partitioned operators themselves. In what follows, we consider only non-interacting theories. Even in this simple case there is, however, some ambiguity in constructing this evolution: Should one first evolve the operator $\hat{\sO}(t)$ forward in time, and then construct the partition, $\hat{\sO}(t)|_S$? Or partition $\hat{\sO}(0)$ and evolve this forward in time,
$\hat{\sO}|_S(t)$? Contrary to what one may expect, the two approaches are not generally equivalent.
\begin{claim}
$\hat{\sO}(t)|_S = \hat{\sO}|_S(t)$ if and only if $\com{H}{\P_S}=0$.
\end{claim}
\begin{proof}
In the absence of interactions, we have in general that $\hat{U}^{\dagger}\hat{\sO}\hat{U}=\SUM{nm}{}\bra{n}U^{\dagger}OU\ket{m}\hat{c}^{\dagger}_n\hat{c}_m$. Where $\hat{U}$ is the usual evolution operator generated by Hamiltonian $\hat{H}$. Then the partition of the time-evolved operator is
\begin{align}
\hat{\sO}(t)|_S=\frac{1}{2}\SUM{nm}{}\bra{n}\acom{O_H(t)}{\P_S}\ket{m}\hat{c}^{\dagger}_n\hat{c}_m,
\end{align}
with $O_H(t)=U^{\dagger}OU$ the Heisenberg evolution of $O$. On the other hand, the time evolution of the partitioned operator gives
\begin{align} \label{Ooft}
\hat{\sO}|_S(t)= \frac{1}{2}\SUM{nm}{}\bra{n}U^{\dagger}\acom{O}{\P_S}U\ket{m}\hat{c}^{\dagger}_n\hat{c}_m.
\end{align}
Clearly, then, these two expressions will only be equivalent if $\com{H}{\P_S}=0$ so that $U^{\dagger}\P_SU=\P_S$.
\end{proof}
In order to preserve the equivalence between the Schr\"odinger and Heisenberg pictures, one must use \cref{Ooft} to describe the dynamics of the partitioned observable. We emphasize that in general $\com{H}{\P_S}\neq0$, and so contrary to what one may expect, the evolution of the partition is not given by the partitioning of $\hat{\sO}(t)$.

The operator $\sO|_S$ therefore obeys the inhomogeneous continuity equation
\begin{align} \label{cont}
\frac{d}{dt}\sO|_S=J_{\sO}|_S+\Sigma_{\sO}|_S.
\end{align}
With 
\begin{align}
J_{\sO}|_S&=\frac{1}{2}\acom{O_H}{J|_S}, \label{JA}\\
J|_S&=i\com{H}{\P_S}\eqtxt{,}{and} \label{J}\\
\Sigma_{\sO}|_S&=\frac{1}{2}\acom{\frac{d}{dt}O_H}{\P_S}.\label{Sigma}
\end{align}
In the above, the time dependence of the projection operator $\P_S(t)=U^{\dagger}\P_SU$ has been suppressed. Using the fact that the second quantization of a commutator is the commutator of the second quantized operators, it can readily be shown that $\hat{J}|_S=\frac{d}{dt}\hat{N}_S$, where $\hat{N}_S$ is the number operator for subsystem $\sF{S}$, implying that $J|_S$ is a probability current operator. Thus, we interpret $J_{\sO}|_S$ as a current %flux 
operator for transport of $\sO$ into the subspace $\sH_S$. \cref{Sigma} can be written $\Sigma_\sO|_S=\left[\frac{d}{dt}\sO\right]|_S$, and so we interpret this term as the local production of $\sO$ within the subspace $\sH_S$. 
The Fock-space operator %observable 
$\hat{\sO}|_S$ therefore obeys a similar continuity equation
\begin{align} \label{contflux}
\frac{d}{dt}\hat{\sO}|_S=\hat{J}_{\sO}|_S+\hat{\Sigma}_{\sO}|_S.
\end{align}

As an illustrative application of this partition, consider the energy current passing through site $\ket{i}$ in the %non-interacting 
tight-binding model
\begin{align}
    \hat{H}=-t\SUM{i,\delta}{}\hat{c}^{\dagger}_{i}\hat{c}_{i+\delta}+\mbox{h.c.},
\end{align}
where the sum on $\delta$ denotes a sum over nearest neighbors. Using \cref{JA}, we find for the energy current %flux
\begin{align}
\hat{J}_H|_{\ket{i}}=\frac{it^2}{2}\SUM{\delta,\delta'}{}\of{\hat{c}^{\dagger}_{i+\delta+\delta'}\hat{c}_{i}-\hat{c}^{\dagger}_{i}\hat{c}_{i-\delta-\delta'}},
\end{align}
which implies the textbook definition \cite{Mahan} for the energy current operator:
\begin{align}
\hat{\vec{J}}_H|_{\ket{i}}=-\frac{it^2}{2}\SUM{\delta,\delta'}{}\of{\vec{\delta}+\vec{\delta}'}c^{\dagger}_{i+\delta+\delta'}c_i.
\end{align}

\section{Density of One-Body Observables}

Up until now, we have discussed only partitions over a discrete subspace of $\sH$. The same construction can be applied also to subsets $S\subset\mathbb{R}^n$ by replacing the projector $\P_S$ with the operator $\ket{x}\bra{x}$ for $x\in \mathbb{R}^n$. We thus define the density of a one-body observable $\hat{\sO}$ in much the same way as before:
\begin{align}
    \rho_{\sO}(x)=\frac{1}{2}\acom{O}{\ket{x}\bra{x}}
\end{align}
and
\begin{align} \label{hdensity}
    \hat{\rho}_{\sO}(x)=\SUM{nm}{}\bra{n}\rho_{\sO}(x)\ket{m}\hat{c}^{\dagger}_n\hat{c}_m.
\end{align}
Note that, as with the discrete partition, $\hat{\sO}$ need not be diagonal in position representation to define $\hat{\rho}_{\sO}(x)$. Rather, we argue that 
$\hat{\rho}_{\sO}(x)$ describes the influence of the non-local observable $\hat{\sO}$ at the location $x\in\mathbb{R}^n$. 

This interpretation may be clarified by considering the operator $\hat{\sO}$ in position representation
\begin{align}
    \hat{\sO}&=\int dx\,dy \,\hat{\sO}(x,y)\eqtxt{,}{with} \\
    \hat{\sO}(x,y)&=\SUM{nm}{}\left[O_{nm}\psi_n(x)\psi^*_m(y)\right]\hat{\psi}^{\dagger}(x)\hat{\psi}(y),
\end{align}
where $\hat{\psi}^{\dagger}(x),\hat{\psi}(x)$ are the usual fermionic field creation and annihilation operators, and $\psi_n(x)=\braket{x|n}$. The density in \cref{hdensity} can be rewritten in terms of $\hat{\sO}(x,y)$ as
\begin{align} \label{hrho}
    \hat{\rho}_{\sO}(x)=\frac{1}{2}\int dy\of{\hat{\sO}(x,y)+\hat{\sO}(y,x)}.
\end{align}

%This expression should be thought of in analogy with the classical Coulomb potential, where the local potential is defined by integrating over the charge density:
%\begin{align}
%V=&\int dx \,V(x)\rho(x)\eqtxt{,}{with}\\
%V(x)=&\int dy\frac{\rho(y)}{|x-y|}.
%\end{align}
%The expression $V(x)\rho(x)$ can be interpreted as a local energy density arising from interactions between charges at $x$ and the rest of the distribution. \cref{hrho} should be thought of in the same way: with the local operator density arising from the couplings between the point $x$ and all other points $y\in\mathbb{R}^n$. This will be discussed in more detail within the context of the quantized Coulomb potential in Sec.\ \ref{sec:many_body}.

Following the same arguments as in the previous section, we find that the density operator obeys a continuity equation
\begin{align} \label{contrho}
    \frac{d}{dt}\hat{\rho}_{\sO}(x,t)=-\nabla\cdot \hat{\vec{j}}_{\sO}(x,t)+\hat{\sigma}_{\sO}(x,t),
\end{align}
where
\begin{align}
    \vec{j}_{\sO}(x,t)&=\frac{1}{2}\acom{O_H}{\vec{j}(x,t)}, \label{jO}\\
    \vec{j}(x,t)&=\frac{1}{2m}U^\dagger\acom{\vec{p}}{\ket{x}\bra{x}}U,\label{j}\\
    \sigma_{\sO}(x,t)&=\frac{1}{2}\acom{\frac{d}{dt}O_H}{U^{\dagger}\ket{x}\bra{x}U}.
\end{align}
%[TO DO: Verify the factors of $U^\dagger$ and $U$ added in Eq.\ \eqref{j}.]
In the above, $\vec{p}=-i\nabla$ is the usual momentum operator. Moreover, in slight contrast to \cref{J,Sigma}, we have made explicit the time evolution of the projection operator $U^{\dagger}\ket{x}\bra{x}U$. The definitions of current density in \cref{jO,j} are consequences of the following claim:
\begin{claim}
$i\com{H}{\ket{x}\bra{x}}=-\nabla\cdot \vec{j}(x)$, where $\vec{j}(x)=\frac{1}{2m}\acom{\vec{p}}{\ket{x}\bra{x}}$.
\end{claim}
\begin{proof} %\hspace{1cm}
%\nwln
Only the kinetic part of the Hamiltonian is non-vanishing in the commutator, since for any two states $\ket{n}$ and $\ket{m}$ $$\bra{n}\com{V}{\ket{x}\bra{x}}\ket{m}=\psi_m(x)V\psi^*_n(x)-\psi^*_n(x)V\psi_m(x)=0,$$ assuming that $V$ acts as a simple multiplication operator in position representation. Then 
\begin{align*}
i\com{H}{\ket{x}\bra{x}}=\frac{i}{2m}\com{\vec{p}^{\,2}}{\ket{x}\bra{x}}.
\end{align*}
To proceed, we consider the matrix elements
\begin{align*}
i\bra{n}\com{\vec{p}^{\,2}}{\ket{x}\bra{x}}\ket{m}&=\nabla\cdot\of{\psi_m\of{i\nabla\psi_n}^*+\psi_n^*\of{i\nabla\psi_m}}\\
&=-\nabla\cdot\bra{n}\acom{\vec{p}}{\ket{x}\bra{x}}\ket{m}.
\end{align*}
Therefore, for any two states $\ket{n}$ and $\ket{m}$, $\bra{n}i\com{H}{\ket{x}\bra{x}}\ket{m}=-\frac{1}{2m}\nabla\cdot\bra{n}\acom{\vec{p}}{\ket{x}\bra{x}}\ket{m}$, which proves the claim.
\end{proof}
Note that $\bra{\psi}\vec{j}(x,t)\ket{\psi}$ gives the usual probability current density of the state $\ket{\psi}$.

\section{Entropy Partitions}
Of particular interest is the partition of the von Neumann entropy. Define in the Schr\"odinger picture the entropy operator
\begin{align}
    \hat{S}(t)=-\ln{\hat{\rho}(t)},
\end{align}
where $\hat{\rho}(t)$ is the density matrix for the global ensemble. Then $S(t)=\braket{\hat{S}(t)}=S(0)$ due to the unitary evolution of $\hat{\rho}$. The Heisenberg picture entropy operator is then $\hat{S}_H(t)=-\hat{U}^{\dagger}\ln{\hat{\rho}(t)}\hat{U}=\hat{S}(0)$. So the entropy operator is constant, in agreement with our expectation that the global entropy should be constant under unitary evolution.

In the absence of inter-particle interactions, the state of the system may be taken to be a product of the form
\begin{align} \label{rho}
    \hat{\rho}=\PROD{\ket{k}}{}\left[f_k\hat{c}^{\dagger}_k\hat{c}_k+(1-f_k)\hat{c}_k\hat{c}^{\dagger}_k\right],
\end{align}
where $\left\{\ket{k}\right\}$ is any set of single-particle orbitals which span $\sH$ and $f_k$ describes the probability that the state $\ket{k}$ is occupied. We define the {\it statistical basis} to be the set of single-particle orbitals over which the density matrix factorizes. 
Such a state may describe, for example, a quantum system initially in equilibrium that is subsequently acted upon by a time-dependent external force.
%a pure state which is then coupled to any number of thermal reservoirs at different temperatures and chemical potentials. 

For such a product state, the entropy operator becomes a sum of particle- and hole-ordered single-particle operators
\begin{align}
    \hat{S}=-\SUM{k}{}\ln{f_k}\hat{c}^{\dagger}_k\hat{c}_k-\SUM{k}{}\ln{1-f_k}\hat{c}_k\hat{c}^{\dagger}_k.
\end{align}
The arguments of the previous sections hold just as well for a hole-ordered operator, and so the entropy may be partitioned as in \cref{partition,hatpartition}. In particular, its partition and density obey the continuity equations
\begin{align} 
    &\frac{d}{dt}\hat{S}|_S=\hat{J}_S|_S \label{Scont},\\
    &\frac{d}{dt}\hat{\rho}_S(x)=-\nabla\cdot\hat{\vec{j}}_S(x),\label{Scontdens}
\end{align}
where 
\begin{eqnarray}
\hat{\rho}_S(x) &=& -\frac{1}{2} \sum_{k,k'} \psi_{k'}(x) \psi^\ast_k(x)  \left[ \ln{f_k f_{k'}} c^\dagger_k c_{k'} \right. \nonumber \\
&+& \left. \ln{(1-f_k)(1-f_{k'})}c_{k'} c^\dagger_k \right],
\end{eqnarray}
\begin{eqnarray}\label{eq:vecjs}
\hat{\vec{j}}_S(x) &=& -\frac{1}{2} \sum_{k,k'}  \langle k|\,\vec{j}(x)|k'\rangle \left[ \ln{f_k f_{k'}} c^\dagger_k c_{k'} \right. \nonumber \\
&+& \left. \ln{(1-f_k)(1-f_{k'})}c_{k'} c^\dagger_k \right]
\end{eqnarray}
are the entropy density operator and the entropy current density operator, respectively. %and $\vec{j}_{kk'}(x)=\langle k|\vec{j}(x)|k'\rangle$.
The net entropy current operator %is given by a sum of particle- and hole-like currents 
$\hat{J}_S|_S$ %=\hat{J}^p_S|_S+\hat{J}^h_S|_S$, with
can be obtained as minus the surface integral of Eq.\ \eqref{eq:vecjs}.

%\begin{align}
%    &\hat{J}^p_S|_S=-\frac{1}{2}\SUM{kk'}{}\ln{f_kf_{k'}}\bra{k}J|_S\ket{k'}\hat{c}^{\dagger}_k\hat{c}_{k'}\\
 %   &\hat{J}^h_S|_S=-\frac{1}{2}\SUM{kk'}{}\ln{(1-f_k)(1-f_{k'})}\bra{k}J|_S\ket{k'}\hat{c}_k\hat{c}^{\dagger}_{k'}.
%\end{align}
%The current density is defined similarly.

Crucially, as a consequence of global entropy conservation under unitary evolution, we find that there is no local entropy production in \cref{Scont,Scontdens}. Comparing to \cref{Sredcons}, we interpret this to mean that the entropy partition proposed in this article does not measure entanglement entropy between subsystems. This entropy partition may therefore provide a more faithful description of the local thermodynamic entropy of the system \cite{kumar2024}. 

%Indeed, it can be shown that this entropy current is exactly proportional to a heat current for systems linearly deviating from equilibrium \cite{kumar2024}.

\subsection{Consistency with the Third Law}
Applying the Hilbert space partition to the entropy operator leads to a system entropy operator
\begin{align} \label{Spart}
    \hat{S}|_S=\hat{S}_S+\frac{1}{2}\hat{S}_{SR}.
\end{align}
In light of the above discussion, this would imply that the heat added to an open quantum system should be defined as \cite{kumar2024}
\begin{align}
    \dj Q_S=d\!\braket{\hat{H}_S+\frac{1}{2}\hat{H}_{SR}}-\mu d \!\braket{\hat{N}_S} - \dj W_S
\end{align}
for linear deviations from equilibrium, where $\hat{N}_S$ is the system number operator and $W_S$ is the external work done on the system.

Various authors \cite{Ludovico,EOG,bruchQuantumThermodynamicsDriven2016,haughianQuantumThermodynamicsResonantlevel2018,strasbergFirstSecondLaw2021a,bergmannGreenFunctionPerspective2021,lacerdaQuantumThermodynamicsFast2023} have suggested including different fractions of the coupling energy $\braket{\hat{H}_{SR}}$ in the heat partition. However, it can be shown that any partition of the entropy apart from \cref{Spart}, corresponding to the symmetric partition advocated in Refs.\ \cite{Ludovico,bruchQuantumThermodynamicsDriven2016,haughianQuantumThermodynamicsResonantlevel2018,kumar2024}, leads to a violation of the Third Law of Thermodynamics. Define the $\alpha$-partition of $\hat{S}$ as 
\begin{align}
    \hat{S}^{\alpha}|_S=\hat{S}_S+\alpha\hat{S}_{SR}.
\end{align}
\begin{claim}
Let $\hat{\rho}$ be a product state of the form \cref{rho}. Then for any $\alpha\neq1/2$,
$\braket{\hat{S}^{\alpha}|_S}$ diverges in the limit $f_k\rightarrow0,1$ for any k.
\end{claim}
\begin{proof}
Define the operators in first quantization $S^{n}=-\SUM{k}{}\ln{f_k}\ket{k}\bra{k}$ and $S^{p}=-\SUM{k}{}\ln{1-f_k}\ket{k}\bra{k}$ so that $\hat{S}=\hat{S}^n+\hat{S}^p$ with $\hat{S}^n$ and $\hat{S}^p$ being their corresponding particle- and hole-like quantizations. We will focus our attention on $S^n$, anticipating that the treatment for $S^p$ will be identical, and drop the superscripts.

The coupling term can therefore be written as the second quantization of the operator
\begin{align}
S_{SR}&=\P_SS\P_R+\P_RS\P_S\\
\implies S^{\alpha}|_S&=\P_S S\P_S+\alpha\of{\P_SS\P_R+\P_RS\P_S}
\end{align}
If we write $1=2\alpha+(1-2\alpha)$ and make use of $\P_S+\P_R=\id$, then the last line becomes
\begin{align} \label{eq5}
S^{\alpha}|_S=(1-2\alpha)\P_SS\P_S+\alpha\acom{\P_S}{S}
\end{align}
Consider now the matrix elements in the statistical basis, 
$\left\{\ket{k}\right\}$\footnote{$S$ is diagonal in this basis.}:
\begin{align}
\bra{k}S^{\alpha}|_S\ket{k}=(1-2\alpha)\SUM{k'}{}s_{k'}|\bra{k}\P_S\ket{k'}|^2+2\alpha s_k\bra{k}\P_S\ket{k},
\end{align}
where $s_k=-f_k\ln{f_k}-(1-f_k)\ln{1-f_k}$.
We find, then, for the partitioned entropy
\begin{widetext}
\begin{align} \label{wide}
\braket{\hat{S}^{\alpha}|_S}=&(1-2\alpha)\SUM{kk'}{}\of{-f_k\ln{f_{k'}}-(1-f_k)\ln{1-f_{k'}}}|\bra{k}\P_S\ket{k'}|^2+2\alpha\SUM{k}{}s_k\bra{k}\P_S\ket{k}.
\end{align}
\end{widetext}
 The first term in Eq.\ \eqref{wide} diverges if any $f_k\in\{0,1\}$, unless $\alpha=1/2$ or $\bra{k}\P_S\ket{k'}=\delta_{kk'}$. The latter condition applies only to a partition of the statistical basis, and is certainly not true in general. Thus, we conclude that for a general partition of the entropy of a partially pure state,\footnote{We define a partially pure state to be any state of the system containing definite occupancies, $f_k=0,1$.} $\braket{\hat{S}^{\alpha}|_S}$ diverges if $\alpha\neq1/2$.
\end{proof}

From \cref{eq5} we conclude that the only well defined partition $S^{\alpha}|_S$ is $S^{1/2}|_S=\frac{1}{2}\acom{\P_S}{S}$, the Hilbert-space partition. In particular, this implies that the partitioned entropy of the Gibbs state diverges in the limit $T\rightarrow0$ if $\alpha\neq1/2$, since all states above the Fermi level will be unoccupied with definite probability. We conclude that the Hilbert-space partition is the only thermodynamically consistent partition of the entropy and heat for an open quantum system in equilibrium with its surroundings.

Note, though, that our conclusion is even stronger than this. From \cref{wide} it is clear that the partitioned entropy will diverge for $\alpha\neq1/2$ if there are \textit{any} localized pure states in the product ensemble \cref{rho}.

\section{Many-Body Operators}
\label{sec:many_body}
Until now, we have considered only partitions for one-body observables. It is, however, also possible to construct a partition for $N$-body observables based on the partition of the single-particle Hilbert space. %, for arbitrary $N$. 
In the remainder of this section, we generalize the above partitions first to $2-$body, and then $N$-body observables.

%In the absence of interactions, they obey the same continuity equation.

%Up until now we have considered only partitions for non-interacting systems. The reason for this is two-fold. Introducing interactions into the ensemble necessitates the partitioning of many-particle operators. In addition, and perhaps more importantly, the system may only be factored into the product state \cref{rho} in the absence of interactions. In general $\hat{\rho}(t)$, and therefore $\hat{S}(t)$ is vastly more complex.

%It is, however, possible to construct a partition of $N$-body observables for arbitrary $N$. A partition of any many-body operator may therefore be constructed as a sum over $N$-body partitions. Owing to the constancy of the Heisenberg picture entropy operator, $\hat{S}_H(t)=\hat{S}_H(0)$, this implies that the entropy may be partitioned even in the presence of interactions, provided that the ensemble was initially described by a product state. We note that this prescription is similar to that of the Keldysh contour in Green's function theory for open quantum systems, in which the system is taken to be in thermal equilibrium with the reservoir in the infinite past [Citations Needed].

\subsection{Two-Body Operators}
Consider first a two-body operator $O=\sum_{ijkl}A_{ijkl}\ket{ij}\bra{kl}$ which acts on the Hilbert space $\sH\otimes\sH$. We wish to generalize the partition of $O$ onto a subspace $\sH_S\subset\sH$ in such a way that $O|_{S}+O|_{S^C}=O$, where $S^C=\SET{\ket{j}\in\sH}{\ket{j}\not\in \sH_S}.$ This can be accomplished by adjusting the projection operator as follows $\mathbb{P}_S\rightarrow\frac{1}{2}\mathbb{P}_S\hat{\oplus}\mathbb{P}_S\equiv\frac{1}{2}\of{\mathbb{P}_S\otimes\id+\id\otimes\mathbb{P}_S}$:
\begin{align}\label{twobodyrho}
O|_{S}\equiv \frac{1}{2}\acom{O}{\frac{\mathbb{P}_S\hat{\oplus}\mathbb{P}_S}{2}}.
\end{align}
Then, because $\of{A\hat{\oplus}B}+\of{C\hat{\oplus}D}=\of{A+C}\hat{\oplus}\of{B+D}$ and $\mathbb{P}_S+\mathbb{P}_{S^C}=\id$, $O|_{S}+O|_{S^C}=\frac{1}{2}\acom{O}{\frac{2\id}{2}}=O$ as needed. It then follows that the partition of the second quantized operator $\hat{\sO}$ is
\begin{align}
    \hat{\sO}|_S=\SUM{ijkl}{}\bra{ij}O|_S\ket{kl}\hat{c}^{\dagger}_{i}\hat{c}^{\dagger}_j\hat{c}_k\hat{c}_l.
\end{align}

One may generalize the density of observables in the same manner, from which we obtain, in position representation,
\begin{align} \label{rho2}
    \hat{\rho}_{\sO}(x)=\frac{1}{4}\int&{dwdydz}[\hat{\sO}(w,x,y,z)+\hat{\sO}(x,w,y,z)\\&+\hat{\sO}(w,y,x,z)+\hat{\sO}(w,y,z,x)]. \nonumber
\end{align}

As motivation for this construction, consider the Coulomb interaction 
\begin{align}
    \hat{V}=\frac{1}{2}\int{dxdy}\frac{\normord{\hat{\rho}(x)\hat{\rho}(y)}}{|x-y|},
\end{align}
where $\hat{\rho}(x)=\hat{\psi}^{\dagger}(x)\hat{\psi}(x)$ and $\normord{\hat{A}}$ denotes the normal ordering of $\hat{A}$. Defining $$\hat{V}(w,x,y,z)=\frac{\delta(w-x)\delta(z-y)}{2|x-y|}\of{\hat{\psi}^{\dagger}(x)\hat{\psi}^{\dagger}(y)\hat{\psi}(w)\hat{\psi}(z)}$$ we see that $\hat{V}=\int{dwdxdydz}\hat{V}(w,x,y,z)$. \cref{rho2} then implies that the energy density is
\begin{align}
    \hat{\rho}_V(x)&=\normord{\hat{V}(x)\hat{\rho}(x)}\eqtxt{,}{with}\\
    \hat{V}(x)&=\frac{1}{2}\int{dy}\frac{\hat{\rho}(y)}{|x-y|},
\end{align}
in analogy with the classical description.

\subsection{N-Body Observables}
For the matrix elements of an $N$-body operator $O\in\sL\of{\sH^N}$ we define the partition
\begin{align}\label{ON}
    O|_S=\frac{1}{2}\acom{O}{\frac{1}{{N}}\hat{\bigoplus}_{i=1}^{N}\P_S}.
\end{align}

\begin{claim}
    For any $N$-body observable $\hat{\sO}$, its partition over $\sH_S\subset\sH$ is given by
    \begin{align}\label{best}
        \hat{\sO}|_S&=\frac{1}{2N}\acom{\hat{\sO}}{\hat{N}_S}-\frac{1}{N}\normord{\hat{\sO}\hat{N}_S},\\
        \hat{N}_S&=\SUM{\ket{n}\in\sH_S}{}\hat{c}^{\dagger}_n\hat{c}_n. \nonumber
    \end{align}
\end{claim}
\begin{proof}
    The proof follows straightforwardly from an application of Wick's theorem to the anticommutator $\frac{1}{2N}\acom{\hat{\sO}}{\hat{N}_S}$. Upon comparing the normal ordering to the second quantization of \cref{ON}, one can see that they differ by exactly $\frac{1}{N}\normord{\hat{\sO}\hat{N}_S}$.
\end{proof}
\begin{cor}
The partition of any observable $\hat{\sO}|_S$ obeys the continuity equation in \cref{contflux}, even in the presence of interactions.
\end{cor}

Following the same arguments, one can show that the density of many-body observables is similarly defined.
\begin{claim}
    The density of any $N$-body observable $\hat{\sO}$ may be defined as
    \begin{align}
        \hat{\rho}_{\sO}(x)=\frac{1}{2N}\acom{\hat{\sO}}{\hat{\rho}(x)}-\frac{1}{N}\normord{\hat{\sO}\hat{\rho}(x)}.
    \end{align}
\end{claim}
%[TO DO: Please verify the 1/N factor in the above equation.]
%\begin{cor}
%    The density $\hat{\rho}_{\sO}(x,t)$ obeys the conservation law \cref{contrho}, even in the presence of interactions.
%\end{cor}

\section{Conclusions}

In this article, we have developed a framework to partition quantum observables based on a partition of the underlying single-particle Hilbert space.  We have provided explicit expressions for the partition of Fock-space operators corresponding to generic $N$-body observables.  For the case of an open quantum system, a bipartite partition between system (S) and reservoir (R) was applied to both the Hamiltonian and the entropy of the system.  The Hilbert-space partition was shown to correspond to a symmetric ($\alpha=1/2$) partition \cite{Ludovico,bruchQuantumThermodynamicsDriven2016,haughianQuantumThermodynamicsResonantlevel2018,kumar2024} of the off-diagonal 1-body observables, such as the coupling Hamiltonian $H_{SR}$ between system and reservoir.  Any partition %\cite{EOG} 
with $\alpha\neq 1/2$ was shown to lead to a singular entropy, and hence does not provide a basis to construct a consistent thermodynamics from the statistical mechanics of the problem.

\acknowledgements
We thank Carter Eckel, Ferdinand Evers, Marco Jimenez, Parth Kumar, and Yiheng Xu for useful discussions.

\bibliographystyle{apsrev4-1} % Tell bibtex which bibliography style to use
\bibliography{Bibs/main} % Tell bibtex which .bib file to use (this one is some example file in %TexLive's file tree)

\end{document}